\documentclass[%
    preprint,
nofootinbib,
amsmath,amssymb,
 aps,
prd,
]{revtex4-2}

\usepackage{amsmath}
\usepackage{amsfonts}
\usepackage{amssymb}
\usepackage{amsthm}
\usepackage{mathrsfs}
\usepackage{mathtools}
\usepackage{natbib}
\usepackage{tikz}
\usepackage{MnSymbol}
\usepackage{slashed}
\usepackage{graphicx}
\usepackage{svg}
\usepackage{orcidlink}
\usepackage{bm}
\usepackage[T2A]{fontenc}
\usepackage[utf8]{inputenc} % Only needed for older engines; modern ones use it by default
\usepackage[russian,english]{babel}

\svgsetup{inkscapelatex=false}

\numberwithin{equation}{section}




\begin{document}

%%% Theorem Style %%%
\newtheoremstyle{theorem}
                 {5pt}
                 {5pt}
                 {}
                 {}
                 {\bfseries}
                 {\newline}
                 {0pt}
                 {}

\theoremstyle{theorem}
\newtheorem{theorem}{Theorem}
\newtheorem{lemma}[theorem]{Lemma}
\newtheorem{definition}[theorem]{Definition}
\newtheorem{proposition}[theorem]{Proposition}
\newtheorem{corollary}[theorem]{Corollary}

\newtheoremstyle{remark}
                 {5pt}
                 {5pt}
                 {}
                 {}
                 {\bfseries}
                 {\newline}
                 {0pt}
                 {}

\theoremstyle{remark}
\newtheorem*{remark}{Remark.}
%%% %%%

\title{Perfect fluids revisited: an action principle approach}

\author{Kostas Tzanavaris\,\orcidlink{0000-0003-0949-7191}}\affiliation{\AEI}\affiliation{\Leibniz}\affiliation{\Edinburgh}
\newcommand*{\Edinburgh}{Higgs Centre for Theoretical Physics, James Clerk Maxwell Building, Edinburgh EH9 3FD, UK}
\newcommand*{\AEI}{Max Planck Institute for Gravitational Physics (Albert Einstein Institute), D-30167 Hannover, Germany}
\newcommand*{\Leibniz}{Leibniz Universit{\"a}t Hannover, D-30167 Hannover, Germany}

\begin{abstract}
We revisit the variational principle for relativistic perfect fluids in a manifestly covariant formulation based on differential forms, with particular attention to the boundary data required for a well-posed action principle. For timelike flows, the formulation reproduces the standard formulation of Schutz and Brown. We then examine the extension obtained by replacing proper time normalization of the velocity with a nullity constraint, with the he Bi{\v c}{\'a}k-Kucha{\v r} null dust action emerging as a special case. Requiring the complete action to be invariant under local rescalings of the null flow vector imposes the condition that the energy density depend only on the entropy per particle. For this class of null flows, the conventional equilibrium condition, \textit{i.e.}, that the energy density is convex, fails at non-zero temperatures, therefore limiting the interpretation of the fluid as an ordinary thermodynamic system.
\end{abstract}

\maketitle

%\tableofcontents
\section{Introduction}
\label{sec:introduction}

The action principle for relativistic fluids has a long history. Early formulations of Taub \cite{taub1954fluid} and Fock \cite{fock1955teoriya} described the fluid in terms of a congruence of particle worldlines, implementing Hamilton's principle through variations of the individual flow lines. While this Lagrangian viewpoint is physically transparent, it tends to obscure the covariance of the formulation. A major step towards a covariant Eulerian description was achieved by Schutz \cite{schutz1970fluid}, who expressed the fluid velocity in terms of a Clebsch type decomposition, building on earlier work of Seliger and Whitham \cite{seliger1968continuum}. The relation between velocity potential actions, Lagrangian coordinates and constrained variational principles was further clarified in subsequent work by Schutz, Sorkin, Carter, Brown, and others \cite{schutz1977fluid,1991RSPSA.433...45C,brown1993fluid}. While variational principles for null dust already exist, notably the formulation of Bi{\v c}{\'a}k and Kucha{\v r} \cite{Bicak:1997bx}, there isn't any clear generalisation beyond null dust that connects to the aforementioned literature for timelike flows. 

The present work revisits the variational description of relativistic fluids in geometrical language for general causal flows, separating the kinematical variables of the flow from the thermodynamic variables of the equation of state. In many standard treatments, the particle current is taken as the fundamental variable, with the particle density and 4-velocity recovered from the norm and direction of this current. This is efficient for timelike flows, but it becomes degenerate when the current is null. We instead take the velocity, particle number, and entropy density as the independent variables, subject to constraints. Brown’s formulation \cite{brown1993fluid} is very close to the present construction as it also admits a description in which the density and velocity are varied independently, subject to normalization. Here, we use differential forms to make the dependence on the metric tensor explicit. This allows the same variational ansatz to be applied simultaneously for timelike and null flows by simply replacing the usual normalization condition $u_\alpha u^\alpha = -1$ by the condition 
\begin{equation}
    \label{eqn:vel-cond}
    u_\alpha u^\alpha = c, \quad c=0,-1. 
\end{equation}
The matter equations for the fluid then imply 
\begin{equation}
    \rho + P = cf,
\end{equation}
where $\rho$ and $P$ are identified with the energy density and pressure in the case of timelike flows, and $f$ is the Lagrangian multiplier associated with the condition \eqref{eqn:vel-cond}. In the case of timelike flows $(c=-1)$, we reproduce the standard formulations of Schutz \cite{schutz1970fluid,schutz1977fluid} and Brown \cite{brown1993fluid}. In the case of null flows $(c=0)$ one obtains the constraint $\rho + P =0$, and the resulting stress-energy tensor results in an intermediate between a ``cosmological constant-like source'' and a null dust contribution, thus generalising the Bi{\v c}{\'a}k-Kucha{\v r} formulation. 

The paper is organised as follows. Section \ref{sec:perfect-fluids} reviews the thermodynamic and kinematical variables used throughout the work and introduces the Lagrangian coordinates associated with a causal flow. Section \ref{sec:action-principle} formulates the action principle and derives the matter equations, with particular emphasis on the boundary terms and the corresponding Dirichlet data. Section \ref{sec:action-symmetries} discusses the continuous internal symmetries of the action and the associated conserved currents. Section \ref{sec:stress-energy} derives the stress-energy tensor for general (causal) flows by varying the co-frame. Finally, the equation of state for null flows is discussed in Section \ref{sec:null-flows} and an example of a non-dust null flow is constructed in Section \ref{sec:null-example}. Appendix \ref{sec:lag-coord-proof} contains a proof of the local existence of Lagrangian coordinates, and Appendix \ref{sec:tetrad-formalism} a summary of the differential form identities used in the tetrad formulation of general relativity.

\vspace{0.5em}

\noindent
\textbf{Notation and units.} Throughout this paper, the pair $(\mathscr{M},g)$ will denote a four-dimensional Lorentzian manifold. We use the ``mostly positive'' MTW convention for the signature $(-,+,+,+)$. The set of smooth differential forms of degree $k$ is denoted by $\Omega^k(\mathscr{M})$, and the space of smooth differential forms on $\mathscr{M}$ is denoted by $\Omega(\mathscr{M})$, and the image of every vector field $X$ under a 1-form $\omega$ is denoted by $\langle\omega,X\rangle$. Given any vector field $X$ on $\mathscr{M}$, its metrically equivalent 1-form is defined by $X^\flat=g(X,\cdot)$. Similarly, for every $\omega\in\Omega^1(\mathscr{M})$ we define its metrically equivalent vector field as the (unique) vector field $\omega_\sharp$ satisfying $\omega = g(\omega_\sharp,\cdot)$. The inverse metric tensor is defined as the tensor field $g^\sharp$ on $T^*\mathscr{M}\otimes T^*\mathscr{M}$ whose components $g^{\alpha\beta}$ in a local frame are the solutions to the equation $g^{\alpha\beta}g_{\beta\gamma} = {\delta^\alpha}_\gamma$. Throughout this paper, the natural system of units is used $(G=k_B=c=1)$.

\section{Perfect fluids}
\label{sec:perfect-fluids}

In this section, we outline the basic properties of perfect fluids, following \cite{Andersson:2006nr}. The independent variables describing the fluid under study are taken to be the particle number density and entropy density.

\subsection{Thermodynamics}
\label{sec:thermodynamics}

Perfect fluids are idealized thermodynamic systems, characterized by zero viscosity and shear stresses. For simplicity, we consider a one-component fluid. For timelike flows, the thermodynamic quantities are measured in the fluid's local rest frame. We further assume that the fluid's motion under consideration is described by a quasi-static process. Its thermodynamic state is then described by the 1st law of thermodynamics
\begin{equation}
    \label{eqn:1st-law-ext}
    dE = T\,dS - P\,dV + \mu\,dN,
\end{equation}
where $E$ is the internal energy, $T$ is its temperature, $P$ is its pressure, $V$ is the volume of the parcel under consideration, $\mu$ is the chemical potential, and $N$ is the total particle number. The equation of state for this system $E=E(S,V,N)$ is homogeneous, and therefore these variables are related by the Gibbs-Duhem relation (\textit{c.f.} \cite{LandauLifshitzStatPhys1}, $\S$24): 
\begin{equation}
    \label{eqn:gibbs-duhem-ext}
    E=TS-PV + \mu N. 
\end{equation}
One often describes the thermodynamic state of perfect fluids in terms of the \textit{energy, entropy and particle number densities}, defined by 
\begin{equation}
    \rho = E/V, \quad s = S/V, \quad n =N/V,
\end{equation}
respectively. In terms of these parameters, the relations \eqref{eqn:1st-law-ext} and \eqref{eqn:gibbs-duhem-ext} reduce to 
\begin{equation}
    \label{eqn:therm-def}
    d\rho = T\,ds + \mu\,dn, \quad
    dP = s\,dT + n\,d\mu, \quad 
    h = Ts + \mu n
\end{equation}
where $h:=\rho + P$ is the \textit{enthalpy density} of the system. 

In much of the literature, the thermodynamic state is described using the particle number density $n$ and the entropy per particle $s/n$. Here, we use $n$ and $s$ as independent scalar variables. As we shall see, the two parametrisations are equivalent for timelike flows. For the extension to null flows, we retain a prescribed constitutive function $\rho=\rho(n,s)$ with the remaining thermodynamic parameters $T,\mu,P$ defined by \eqref{eqn:therm-def}. Their physical interpretation however is not straightforward due to the absence of a local rest frame, and is discussed later.

\subsection{Fluid kinematics}
\label{sec:kinematics}

There are two complementary ways of describing the kinematics of a fluid. In the \textit{Eulerian description} one assigns to each event $p\in \mathscr{M}$ a vector $u_p\in T_p\mathscr{M}$, representing the velocity of a fluid parcel at $p$. For an ordinary relativistic fluid, $u$ is future-directed and unit timelike. In the present work, we allow $u$ to be causal and non-zero, with either the unit timelike normalization condition or the nullity condition imposed later as a constraint in the action.

The integral curves of $u$ are the \textit{flow lines} of the fluid, and represent the worldlines of individual parcels comprising the fluid. In the \textit{Lagrangian} or \textit{co-moving description}, one labels these flow lines by three scalar fields $\mathrm{a}^1,\mathrm{a}^2,\mathrm{a}^3$. These labels are constant along the flow, \textit{i.e.}, $\pounds_u\mathrm{a}^i = 0$ for all $i=1,2,3$. Locally, one may supplement these labels by a parameter $\tau$ along the flow satisfying $\pounds_u\tau = 1$. For timelike flows, this parameter may be chosen to coincide with the proper time of the fluid parcels. 

\begin{theorem}[Local existence of Lagrangian coordinates]
\label{thm:lag-existence}
Suppose $(\mathscr{M},g)$ is a four-dimensional space-time, $u$ is a $C^1$ causal and non-zero vector field on $\mathscr{M}$, and $x^1,x^2,x^3$ are coordinates of a sufficiently small spacelike hypersurface $S\subset\mathscr{M}$. Then, the functions $\mathrm{a}^1,\mathrm{a}^2,\mathrm{a}^3,\tau$ defined by 
\begin{subequations}
\begin{align}
    \label{eqn:lag-cord-def}
    & \pounds_u\mathrm{a}^i=0, 
    &\pounds_u\tau =1,
    && \\
    & \mathrm{a}^i\rvert_S = x^i, 
    & \tau\rvert_S = 0,  
    && 1\leq i\leq 3,
\end{align}    
\end{subequations}
where $\pounds_u$ is the Lie derivative along $u$, form a local coordinate system of $\mathscr{M}$ containing $S$.
\end{theorem}

Theorem \ref{thm:lag-existence} is quite similar to the \textit{rectification theorem for vector fields} \cite{arnold1992odes} which states that given a vector field $u$ that is non-zero in an open neighbourhood $\mathscr{U}$, one can always construct a local coordinate system $(\mathscr{V},x)$ in $\mathscr{U}$ such that $u=\partial_{x^1}$ in $\mathscr{V}$. Eq. \eqref{eqn:lag-cord-def} is simply the statement that the Lagrangian coordinates remain constant along the flow of $u$. For a non-relativistic fluid in a local rest frame $(t,x^i)$, we have $u^0\simeq 1$ and eq. \eqref{eqn:lag-cord-def} becomes the \textit{transport equation}
\begin{equation}
    \left(\partial_t+ u^j\partial_{x^j}\right)\mathrm{a}^i = 0, 
    \quad 
    \tau = t, \quad 1\leq i\leq 3.
    \label{eqn:transport}
\end{equation}

\section{The action principle for perfect fluids}
\label{sec:action-principle}

\subsection{Formulation of the action principle}

We now formulate the action principle for fluids along the lines of Fock and Taub. Recall that the action for a massive particle is $\mathcal{S}_\textrm{particle}[\gamma] = -m\int_\gamma d\tau$, where $m$ is the particle's mass, $\tau$ is the particle's proper time, and $\gamma$ is a compact segment of the particle's worldline. A formal sum of the individual actions results in the action functional for pressureless dust. For a general, inviscid fluid, we instead take the energy density to be a function $\rho=\rho(n,s)$ of the particle density $n$ and entropy density $s$. Hence, the bare Lagrangian density $\mathcal{L}_0$ of the fluid is given by
\begin{equation}
    \mathcal{L}_0=-\rho\,\textrm{Vol}_g.
\end{equation}
The total Lagrangian $\mathcal{L}$ of the fluid is obtained by introducing constraints. We start by introducing local Lagrangian coordinates $\mathrm{a}^1,\mathrm{a}^2,\mathrm{a}^3$, labelling the flow lines. These fields are varied in the interior of the action region, while their values are held fixed on its boundary. The requirement that the labels remain constant along the flow is imposed through the constraints $\pounds_u\mathrm{a}^i=0$, or equivalently
\begin{equation}
    \label{eqn:constr-1-vel}
    d\mathrm{a}^i\wedge\star u^\flat = (\pounds_u\mathrm{a}^i)\,\textrm{Vol}_g = 0, \quad 1\leq i\leq 3,
\end{equation}
where $\star$ is the Hodge star operator\footnote{Recall that the Hodge star operator $\star$ on a $n$-dimensional semi-Riemannian manifold $(\mathscr{M},g)$ is the unique linear operator on $\Omega(\mathscr{M})$ which maps every $k$-form $\omega$ to a $(n-k)$-form $\star\omega$ such that \[\xi\wedge\star\omega=g^\sharp(\xi,\omega)\,\textrm{Vol}_g,\] for all $\xi\in\Omega^k(\mathscr{M})$. An alternative definition in terms of a local frame is given by \eqref{eqn:hodge-star-def}.} corresponding to the metric $g$. Similarly, we impose conservation of particle number and entropy: $\nabla_\alpha(nu^\alpha)=0$ and $\nabla_\alpha(su^\alpha)=0$ respectively. Again, in terms of differential forms, these conservation laws are 
\begin{equation}
    \label{eqn:constr-2-vel}
    \textrm{div}\,(nu^\flat) = 0,\quad  \textrm{div}\,(su^\flat) = 0,
\end{equation}
where $\textrm{div}:\Omega(\mathscr{M})\to\Omega(\mathscr{M})$ is the divergence operator expressed in terms of the Hodge star operator
\begin{equation}
    \textrm{div}\,\omega = (-1)^{k-1}\star^{-1} d \star \omega,
    \quad \omega\in\Omega^k(\mathscr{M}). 
\end{equation}
Lastly, we impose
\begin{equation}
    g(u,u) = c,\quad c=0,-1. 
\end{equation}
For timelike flows $(c=-1)$, this imposes proper time parametrisation. For $c=0$, it imposes that the flow is null. However, this constraint does not fix the parametrisation; this is analysed in later sections. Thus, under these constraints, the total Lagrangian density of the fluid is given by 
\begin{equation}
    \label{eqn:lag-dens-off-shell}
    \mathcal{L} = -\left(\rho + \mathcal{A} + \langle\mathcal{C},u\rangle\right)\textrm{Vol}_g,
\end{equation}
where 
\begin{subequations}
\begin{align}
    \mathcal{A} & = \frac{1}{2}f\left(g(u,u) - c\right), \label{eqn:affine-param}\\
    \mathcal{C} & = n\,d\Phi + s\,d\Theta + b_i\,d\mathrm{a}^i,
\end{align}    
\end{subequations}
are the terms related to the constraints imposed on the fluid, and $f,\Phi,\Theta,b_i$ are the Lagrange multipliers.

\subsection{Implications of the action principle}

As far as the equations of motion are concerned, the action \eqref{eqn:lag-dens-off-shell} is defined only up to boundary terms. However, boundary terms can affect the action principle by determining which fields (or combinations thereof) are fixed on the boundary. On this basis, we argue that the action \eqref{eqn:lag-dens-off-shell} is the appropriate action for perfect fluids. The non-metric variations of $\mathcal{L}$, which are rather straightforward to calculate, are 
\begin{subequations}
\begin{align}
    \delta_u\mathcal{L} & = -(fu^\flat + \mathcal{C})\wedge\star\delta u^\flat, \label{eqn:euler-first}\\
    \delta_n\mathcal{L} & = -(\mu +\pounds_u\Phi)\delta n\,\textrm{Vol}_g, \\
    \delta_s\mathcal{L} & = -(T + \pounds_u\Theta)\delta s\,\textrm{Vol}_g, \\
    \delta_{\mathbf{a}}\mathcal{L} & = \textrm{div}(b_i u^\flat) \delta\mathrm{a}^i\,\textrm{Vol}_g - d(b_i\delta\mathrm{a}^i\star u^\flat),\\
    \delta_\Phi\mathcal{L} & = \textrm{div}(nu^\flat)\delta\Phi\,\textrm{Vol}_g - d(n\delta\Phi\star u^\flat), \\
    \delta_\Theta\mathcal{L} & = \textrm{div}(su^\flat)\delta\Theta\,\textrm{Vol}_g - d(s\delta\Theta\star u^\flat),\\
    \delta_{\bm{b}}\mathcal{L} & = -\delta b_i\,\pounds_u a^i\,\textrm{Vol}_g, \\
    \delta_f\mathcal{L} & = -\frac{1}{2}\delta f\,(g(u,u) - c)\,\textrm{Vol}_g. \label{eqn:euler-last}
\end{align}
\end{subequations}
Now, consider the action principle in a bounded region $\mathscr{D}\subset\mathscr{M}$, which states that 
\begin{equation}
    \delta\int_{\mathscr{D}}\mathcal{L} = 0,
\end{equation}
for all non-metric variations. Keeping the bounded region $\mathscr{D}$ fixed, and working with variations supported in the interior of $\mathscr{D}$, we get the bulk equations
\begin{subequations}
\begin{align}
    & fu^\flat = - \mathcal{C}, && g(u,u) = c, 
    \label{eqn:velocity-rel}\\
    & \mu = -\pounds_u\Phi, && T = -\pounds_u\Theta, 
    \label{eqn:therm-pot}\\
    & \textrm{div}(nu^\flat) = 0, &&
    \textrm{div}(su^\flat) = 0,
    \label{eqn:cons-1}\\
    & \pounds_u\mathrm{a}^i = 0, && \textrm{div}(b_iu^\flat) = 0, && 1\leq i\leq 3,
    \label{eqn:cons-2}
\end{align}
Using \eqref{eqn:therm-pot} in conjunction with the Gibbs-Duhem relation, the relations \eqref{eqn:velocity-rel} reduce to
\begin{equation}
    \label{eqn:enthalpy-constraint}
    h = \rho + P = cf. 
\end{equation}
\end{subequations}
Now, consider arbitrary compactly supported variations that do not vanish at the boundary $\partial\mathscr{D}$. The action principle for such variations necessitates that one fixes the boundary data for $\Phi,\Theta$ and $\mathrm{a}^i$, or equivalently, setting the boundary variations to zero, 
\begin{equation}
    \delta\Phi\rvert_{\partial\mathscr{D}} = \delta\Theta\rvert_{\partial\mathscr{D}} = 
    \delta\mathrm{a}^i\rvert_{\partial\mathscr{D}} = 0.
\end{equation}
In other words, the action principle for the Lagrangian density \eqref{eqn:lag-dens-off-shell} is equivalent to a Dirichlet-type problem for the potentials $\Phi,\Theta$, and fixing the positions of the fluid parcels at the boundary. 

\section{Symmetries of the action}
\label{sec:action-symmetries}

The action contains a redundancy in the parametrisation of the velocity 1-form $u^\flat$. The physical fields $u,n,s$ are determined by the equations of motion, while the potentials $\Phi,\Theta$, the Lagrangian labels $\mathrm{a}^i$ and the multipliers $b_i$ can only enter through the 1-form $\mathcal{C}$. Consequently, any field redefinition that leaves $n,s,u$ and $\mathcal{C}$ invariant is an internal symmetry of the action. We consider infinitesimal and ultralocal transformations of the variables $\Phi,\Theta,\mathrm{a}^i$ and $b_i$, leaving $n$ and $s$ fixed. Let $X$ be the corresponding vector field on this internal space,  
\begin{equation}
    X = X^{\Phi}\frac{\partial}{\partial \Phi}
    + X^\Theta\frac{\partial}{\partial \Theta}
    + X^{b_i}\frac{\partial}{\partial b_i}
    + X^{\mathrm{a}^i}\frac{\partial}{\partial \mathrm{a}^i}.
\end{equation}
For $X$ to generate a symmetry of the system, its flow preserves $\mathcal{C}$:
\begin{equation}
    \label{eqn:redef-cond}
    \pounds_X\mathcal{C} = 0. 
\end{equation}
Define 
\begin{equation}
    \mathcal{F} = \langle\mathcal{C},X\rangle = n X^{\Phi} + s X^{\Theta} + b_i X^{\mathrm{a}^i},
\end{equation}
as the generating function of the redefinition. Then, \eqref{eqn:redef-cond} reduces to 
\begin{equation}
    \label{eqn:redef-cond-2}
    d\mathcal{F} - X^\Phi\,dn - X^\Theta\,ds + X^{b_i}\,d\mathrm{a}^i - X^{\mathrm{a}^i}\,db_i = 0.
\end{equation}
Expanding $d\mathcal{F}$ and comparing coefficients, we obtain
\begin{equation}
    \partial_\Phi\mathcal{F} = \partial_\Theta\mathcal{F} = 0,
\end{equation}
and therefore $\mathcal{F}=\mathcal{F}(n,s,b_i,\mathrm{a}^i)$. The components of the generator $X$ are then given by
\begin{equation}
    X^\Phi = \frac{\partial \mathcal{F}}{\partial n},\;
    X^\Theta = \frac{\partial \mathcal{F}}{\partial s},\;
    X^{\mathrm{a}^i} = \frac{\partial \mathcal{F}}{\partial b_i},\;
    X^{b_i} = -\frac{\partial \mathcal{F}}{\partial \mathrm{a}^i}.
\end{equation}
Hence the infinitesimal symmetries of the action are
\begin{subequations}
\begin{align}
    \delta\Phi & := X^{\Phi} = \partial_n\mathcal{F},\\
    \delta\Theta & := X^{\Theta} = \partial_s\mathcal{F},\\
    \delta\mathrm{a}^i & := X^{\mathrm{a}^i} = \partial_{b_i}\mathcal{F},\\
    \delta b_i & := X^{b_i} = -\partial_{\mathrm{a}^i}\mathcal{F},
\end{align}
\end{subequations}
with $\delta n=\delta s=0$. However, this also implies that $\mathcal{F}$ is further constrained by the relation
\begin{equation}
    \mathcal{F} = n\frac{\partial\mathcal{F}}{\partial n} + s\frac{\partial\mathcal{F}}{\partial s} + b_i\frac{\partial\mathcal{F}}{\partial b_i}. 
\end{equation}
Consequently, $\mathcal{F}$ is a homogeneous function of degree 1 when the Lagrangian coordinates $\mathrm{a}^i$ are fixed, \textit{i.e.} 
\begin{equation}
    \mathcal{F}(\lambda n,\lambda s,\lambda b_i,\mathrm{a}^i) = \lambda \mathcal{F}(n,s,b_i,\mathrm{a}^i)\quad\textrm{for every $\lambda$.}
\end{equation}
For example, setting $\lambda=1/n$ in the above equation, it follows that there is a function $\mathcal{G}$ defined on some subset of $\mathbb{R}^7$ such that 
\begin{equation}
    \mathcal{F}(n,s,b_i,\mathrm{a}^i) = n\mathcal{G}\left(\frac{s}{n},\frac{b_i}{n},\mathrm{a}^i\right).
\end{equation}
Thus, given a vector field $X$ that generates one of the above symmetries of the fluid variables, the induced on-shell variation of the Lagrangian density, in conjunction with \eqref{eqn:euler-first}-\eqref{eqn:euler-last}, is $\delta_X\mathcal{L} = -d(\mathcal{F}\star u^\flat)$, and since $\delta_X\mathcal{L} = 0$ the conserved Noether current is 
\begin{equation}
    {}^X\mathcal{J} = \mathcal{F}\star u^\flat. 
\end{equation}
These symmetries express the freedom to relabel the internal potentials without changing the physical velocity field. They are the covariant analogue of particle relabelling symmetries in the Lagrangian description of fluids. Different choices of the generating function give different conserved charges, transported along the corresponding current.

Lastly, these transformations correspond to the internal symmetries described by Brown \cite{brown1993fluid}. Denoting his variables by a subscript $(B)$, the identification is
\begin{equation}
    s_B=\frac{s}{n},\quad
    \beta_{Bi}=-\frac{b_i}{n},\quad
    \alpha_B^i=\mathrm a^i,\quad
    \varphi_B=-\Phi,\quad
    \theta_B=-\Theta.  
\end{equation}
In particular, Brown’s entropy variable is the entropy per particle. Writing our generating function as $\mathcal{F}=n\mathcal{G}(s/n,b_i/n,\mathrm a^i)$, and identifying Brown’s generating function with
\begin{equation}
    \mathcal{F}_B(s_B,\beta_{Bi},\alpha_B^i) = -\mathcal{G}(s_B,-\beta_{Bi},\alpha_B^i)
\end{equation}
the above transformations reproduce his Eq. (2.18). Thus, the arbitrary function in Brown’s formulation corresponds to the arbitrary function $\mathcal{G}$, while $\mathcal{F}$ satisfies the degree-one homogeneity condition.

\section{The stress-energy tensor}
\label{sec:stress-energy}

For field theories defined by an action principle, their stress-energy tensor is derived by varrying their Lagrangian density with respect to the metric tensor, while keeping all other independent variables fixed. Here, it is more convenient to instead consider variations of the co-frame, keeping the metric tensor coefficients and matter variables fixed. Let $\{e_\alpha\}$ be a local frame on $\mathscr{M}$, and denote by $\{e^\alpha\}$ the corresponding dual co-frame defined by $\langle e^\alpha,e_\beta\rangle = {\delta^\alpha}_\beta$. The the metric tensor expanded in this frame is then given by $g = g_{\alpha\beta}\,e^\alpha\otimes e^\beta$. Under a variation $\{\delta e^\alpha\}$ of the tetrad $\{e^\alpha\}$ given by $\delta e^\alpha = {f^\alpha}_\beta\, e^\beta$, where ${f^\alpha}_\beta$ are compactly supported smooth functions, the induced variations of the metric tensor are 
\begin{equation}
    \delta_{\mathbf{e}} g = 2g_{\alpha\beta}\,\delta e^{(\alpha}\otimes e^{\beta)} = 
    (\delta_{\mathbf{e}} g)_{\alpha\beta}\,e^\alpha\otimes e^\beta,
\end{equation}
where $g_{\alpha\beta}:= g(e_\alpha,e_\beta)$ and $ (\delta_{\mathbf{e}} g)_{\alpha\beta} := g_{\alpha\mu} {f^\mu}_\beta + g_{\beta\mu}{f^\mu}_\alpha$. A brief review of how one calculates variations in this framework is given in Appendix \ref{sec:tetrad-formalism}.

An important property of the Lagrangian for perfect fluids is that it contains no explicit dependence on the connection. For any field whose Lagrangian does not depend on the connection, the stress-energy tensor is derived by 
\begin{equation}
    \label{eqn:def-stress-energy-var}
    \delta_{\mathbf{e}}\mathcal{L} = T_{\alpha\beta}\,\delta e^\alpha\wedge\star e^\beta.
\end{equation}
Assuming that the Lagrangian density is also generally covariant, as is the case for perfect fluids, then one can show that $\nabla_{\alpha}T^{\alpha\beta} = 0$, where $\nabla$ is the Levi-Civita connection. To show that, following an argument of \cite{Noether:1918zz}, we consider an arbitrarily open region $\mathscr{D}\subset\mathscr{M}$ and a smooth vector field $X$ on $\mathscr{M}$ supported in $\mathscr{D}$. We also suppose that $\{e^\alpha\}$ is orthonormal and is defined on all of $\overline{\mathscr{D}}$. Since $X\rvert_{\partial\mathscr{D}}=0$, the flow of $X$ generates diffeomorphisms of $\mathscr{D}$, corresponding to coordinate transformations. Since both the total action and its gravitational part are invariant under these coordinate transformations, the matter action is invariant as well. On solutions of the matter equations, it follows that $\int_{\mathscr{D}}\delta_{\mathbf{e}}\mathcal{L}=0$. On the other hand, because $\delta e^\alpha = \pounds_X e^\alpha$, using Cartan's magic formula and the fact that $\nabla$ is torsion-free and metric, we have 
\begin{align}
    \delta_{\mathbf{e}}\mathcal{L} & = d(T_{\alpha\beta}X^\alpha\star e^\beta) - X^\alpha \left[ d(T_{\alpha\beta}\star e^\beta) - T_{\gamma\beta}{\omega^\gamma}_\alpha\wedge\star e^\beta \right]
    \nonumber\\
    & = d(T_{\alpha\beta}X^\alpha\star e^\beta) - X^\alpha\nabla_\beta {T_\alpha}^\beta\,\textrm{Vol}_g
\end{align}
Integrating over $\mathscr{D}$, we obtain 
\begin{equation}
    \int_{\mathscr{D}}  X^\alpha\nabla_\beta {T_\alpha}^\beta\,\textrm{Vol}_g = 0,
\end{equation}
for every smooth vector field $X$ supported in $\mathscr{D}$. Thus, regardless of the equations for gravity and as long as its Lagrangian is generally covariant, the conservation law for the stress-energy tensor 
\begin{equation}
    \nabla_\alpha T^{\alpha\beta } = 0
\end{equation}
holds independently of the gravitational field equations. Focusing now on fluids, we show that 
\begin{equation}
    \label{eqn:stress-energy-general}
    T_{\alpha\beta} = P g_{\alpha\beta} - fu_\alpha u_\beta.
\end{equation}

\begin{proof}
The variables $\rho$, $u$, $f$ and $\mathcal{C}$ are invariant under variations of the co-frame. However, $u^\flat$, $g$ and consequently $\mathcal{A}$ are not. First, note that $\delta_{\mathbf{e}}\textrm{Vol}_g = g_{\alpha\beta}\delta e^\alpha\wedge\star e^\beta$ and therefore 
\begin{equation}
    \label{eqn:delta-rho}
    \delta_{\mathbf{e}}(\star\rho) = \rho g_{\alpha\beta}\delta e^\alpha\wedge\star e^\beta.
\end{equation}
Similarly, 
\begin{align}
    \delta_{\mathbf{e}}(\mathcal{C}\wedge\star u^\flat) = 
    \delta_{\mathbf{e}}(\langle\mathcal{C},u\rangle\textrm{Vol}_g) = 
    \langle\mathcal{C},u\rangle g_{\alpha\beta}\delta e^\alpha\wedge\star e^\beta.
\end{align}
On-shell, this leads to 
\begin{align}
    \label{eqn:delta-c}
    \delta_{\mathbf{e}}(\mathcal{C}\wedge\star u^\flat) = -(\rho+P)g_{\alpha\beta}\delta e^\alpha\wedge\star e^\beta.
\end{align}
The remaining quantity left to calculate is $\delta_{\mathbf{e}}\mathcal{A}$. By definition, we have $\delta_{\mathbf{e}}\mathcal{A} = \frac{1}{2}f(\delta_{\mathbf{e}} g)(u,u)$, which in turn yields
\begin{align}
    \delta_{\mathbf{e}}\mathcal{A} & = \frac{1}{2}fg_{\alpha\beta}\left[\langle\delta e^\alpha,u\rangle u^\beta + \langle\delta e^\beta,u\rangle u^\alpha\right]
    \nonumber\\
    & = fu_\alpha\langle\delta e^\alpha,u\rangle 
    \nonumber\\
    & = f u_\alpha u^\beta\langle \delta e^\alpha, e_\beta\rangle.
\end{align}
Consequently, because $\mathcal{A}=0$ on-shell, we have 
\begin{equation}
    \label{eqn:delta-A}
    \delta_{\mathbf{e}}(\star\mathcal{A}) = \star\delta_{\mathbf{e}}\mathcal{A} = f u_\alpha u_\beta \delta e^\alpha\wedge\star e^\beta. 
\end{equation}
Combining \eqref{eqn:delta-rho}, \eqref{eqn:delta-c} and \eqref{eqn:delta-A} yields \eqref{eqn:stress-energy-general}.
\end{proof}

\section{Null flows}
\label{sec:null-flows}

Let us now study null flows, \textit{i.e.}), when $(c=0)$. An important distinction from timelike flows is that the nullity condition does not fix the parametrisation. Namely, the condition $u_\alpha u^\alpha=0$ is invariant under the rescaling
\begin{equation}
    \label{eqn:vel-rescale}
    u\longrightarrow\lambda u,
\end{equation}
where $\lambda$ is an arbitrary smooth positive function on $\mathscr{M}$. Preserving the physical currents under this rescaling requires transformations of the other variables characterising the fluid. We propose to require that these transformations preserve the constraint terms and that the complete transformation be a symmetry of the system, \textit{i.e.}, that it preserve the fluid's Lagrangian. As we shall see, this requirement also restricts the equation of state $\rho=\rho(n,s)$.

First, notice that the definition of the Lagrangian coordinates $\mathrm{a}^1,\mathrm{a}^2,\mathrm{a}^3$ given in Theorem \ref{thm:lag-existence} is invariant under the reparametrisation of $u$, and therefore it is natural to impose
\begin{equation}
    \mathrm{a}^i\longrightarrow \mathrm{a}^i, \quad 1\leq i\leq 3,
\end{equation}
under \eqref{eqn:vel-rescale}. Next, consider the particle number and entropy currents 
\begin{equation}
    {}^{(n)}J := nu, \quad {}^{(s)}J := su.
\end{equation}
These are, at least in principle, measurable and must invariant under \eqref{eqn:vel-rescale}. This gives us the transformation rules 
\begin{equation}
    n\longrightarrow \frac{1}{\lambda}n, \quad s\longrightarrow\frac{1}{\lambda}s. 
\end{equation}
Thus, the rest of the constraint variables remain invariant if and only if the transformation rules are 
\begin{equation}
    b_i\longrightarrow \frac{1}{\lambda}b_i, \quad f\longrightarrow\frac{1}{\lambda^2}f, \quad 
    \Phi\longrightarrow\Phi,\quad \Theta\longrightarrow\Theta.
\end{equation}
This transformation is now a symmetry of the system, \textit{i.e.} it preserves its Lagrangian, if and only if 
\begin{equation}
    \rho(n,s) = \rho\left(\frac{n}{\lambda},\frac{s}{\lambda}\right)
\end{equation}
for arbitrary smooth and positive functions $\lambda$ on $\mathscr{M}$. Consequently, there is a function $\varrho$ on $\mathbb{R}$ such that 
\begin{equation}
    \label{eqn:null-state-cond}
    \rho(n,s) = \varrho\left(s/n\right).
\end{equation}

In other words, the energy density of the system is a function of the entropy per particle. Using Maxwell's relations, one can easily see that this is equivalent to requiring that the enthalpy $h=h(n,s)$ vanishes identically. Now, one might've derived the same condition by varying with respect to velocity but there is a conceptual difference between these two. The condition \eqref{eqn:null-state-cond}, which is equivalent to $h=0$, holds \textit{off-shell} and is a requirement of the equation of state, which is itself derived from invariance under \eqref{eqn:vel-rescale}. On the other hand, the condition $h=0$ by itself does not imply symmetry under \eqref{eqn:vel-rescale}. Consider for example the equation of state 
\begin{equation}
    \rho(n,s) = \rho_0 + \frac{1}{2}A(n-n_0)^2 + F(s/n),
\end{equation}
where $\rho_0,A,n_0$ are constants and $F$ is a smooth non-negative function on $\mathbb{R}$. Clearly, $\rho$ is not invariant under the rescaling $(n,s)\to(n/\lambda,s/\lambda)$. However, the condition $h=0$ simply restricts the state space to $n=n_0$, which is also not invariant under rescales. 

Lastly, another consequence of this symmetry requirement is that the stress-energy tensor is independent of the velocity parametrisation, and is thus well-defined. Indeed, given that the resulting stress-energy tensor for null fluids is
\begin{equation}
    T_{\alpha\beta} = -\rho g_{\alpha\beta} - fu_\alpha u_\beta,
\end{equation}
and under $u\to\lambda u$ we have $P\to P$ and $f\to f/\lambda^2$ it follows that $T_{\alpha\beta}\to T_{\alpha\beta}$. The null dust model of Bi{\v c}{\'a}k and Kucha{\v r} \cite{Bicak:1997bx} arises as a special case when $\rho=0$ and constant $\Phi,\Theta$. By eliminating $u$ through $fu^\flat=-b_i\,d\mathrm{a}^i$ and defining $W_i=b_i/\sqrt{-f}$, with $f<0$, the Lagrangian reduces to $\mathcal{L}=-\frac{1}{2}\ell\wedge\star\ell$, where $\ell:=W_i\,d\mathrm a^i$.

\section{A spherically symmetric example}
\label{sec:null-example}

In this section, we produce an example of a spherically solution to the Einstein field equation sourced by the null fluid discussed in the previous section. Suppose that the energy density of the fluid is given by an equation of state $\rho=\rho(n,s)$. Assuming spherical symmetry, we choose a local coordinate system in which the metric tensor assumes the form
\begin{equation}
    g=-F(v,r)\,dv^2 + (dv\otimes dr + dr\otimes dv)+r^2\slashed{g},
\end{equation}
where $\slashed{g}$ is the unit radius metric on the orbit spheres and $\partial_v$ is assumed to be future-pointing. We choose the velocity vector $u$ to also be spherically symmetric and future-pointing. Hence, $u=A\partial_v + B\partial_r$, where $A,B$ are functions of $v,r$. The constraint implies either $A=0$, which describes an ingoing null vector, or $B=FA/2$. Choosing $u$ to be ingoing, its parametrisation is fixed by requiring $g(u,\partial_v)=-1$. Consequently, 
\begin{equation}
    u=-\frac{\partial}{\partial r}.
\end{equation}
Setting $F=1-2m/r$, Einstein's field equations with the stress-energy tensor given by \eqref{eqn:stress-energy-general}, with $P=-\rho$, reduce to 
\begin{equation}
    \frac{\partial m}{\partial r} = 4\pi r^2\rho, \quad 
    \frac{\partial m}{\partial v} = -4\pi r^2 f, \quad 
    \frac{\partial^2 m}{\partial r^2} = 8\pi r\rho. 
\end{equation}
Differentiation of the first equation with respect to $r$ gives us $\partial\rho/\partial r= 0$. Hence, 
\begin{equation}
    m(v,r) = M(v) + \frac{4\pi}{3}\rho(v)r^3,
    \quad 
    f=-\frac{M'(v)}{4\pi r^2} - \frac{1}{3}r\rho'(v),
\end{equation}
where a $'$ denotes differentiation with respect to $v$. A straightforward yet tedious calculation shows that 
\begin{equation}
    \nabla_u u = 0, \quad \nabla_\alpha(fu^\alpha) = -\frac{1}{r^2}\frac{\partial}{\partial r}\left(r^2f\right) = \rho'(v),
\end{equation}
which in turn leads to 
\begin{equation}
    \nabla_\alpha T^{\alpha}{}_\beta = -\partial_\beta\rho - \nabla_\alpha(fu^\alpha)u_\beta 
    -f(\nabla_u u)_\beta = 0,
\end{equation}
since $u_\beta = -\delta^v{}_\beta$. Thus, for the current framework, the stress-energy tensor is automatically divergence-free. 

Up to this point, $M$ and $\rho$ have been arbitrary. We now choose $M$ and $\rho$, as functions of $v$, to be given by 
\begin{equation}
    M(v) = M_0, \quad \rho(v) = \rho_* + qv. 
\end{equation}
where $M_0,\rho_*,q> 0$ are constants. We also choose the equation of state of the system under study to be 
\begin{equation}
    \rho(n,s) = \kappa(s/n)^2,
\end{equation}
where $\kappa>0$ is a constant. 

From here, we can determine the remaining matter variables. For the chosen functions $M$ and $\rho$, Einstein's equation give 
\begin{equation}
    f = -\frac{qr}{3}.
\end{equation}
Then, the equation of state determines the entropy per particle $\sigma$: 
\begin{equation}
    \label{eqn:sigma-v}
    \sigma(v) = \sqrt{\frac{\rho_*+qv}{\kappa}}, \quad \sigma:=\frac{s}{n}. 
\end{equation}
The individual densities are determined by their initial/boundary data. Due to spherical symmetry, the conservation equations for $n$ and $s$ are
\begin{equation}
    \frac{1}{r^2}\partial_r(r^2n) = \frac{1}{r^2}\partial_r(r^2s) = 0.
\end{equation}
A solution compatible with \eqref{eqn:sigma-v} is 
\begin{equation}
    n(v,r)=\frac{n_0(v)}{r^2}, \quad s(v,r)=\frac{n_0(v)}{r^2}\sqrt{\frac{\rho_*+qv}{\kappa}}.
\end{equation}
The chemical potential and temperature are then derived using Maxwell's relations, which in our case are $\mu=\left(\partial\rho/\partial n\right)_{s}$ and $T=\left(\partial\rho/\partial s\right)_n$. These yield 
\begin{equation}
    \mu = -\frac{2\kappa s^2}{n^3}, \quad T=\frac{2\kappa\sigma}{n}.
\end{equation}
Next, the velocity potentials $\Theta,\Phi$ are obtained by integrating \eqref{eqn:therm-pot}. Their general form are 
\begin{equation}
    \Phi(v,r)=\Phi_0(v) - \frac{2\rho(v)r^3}{3n_0(v)}, 
    \quad 
    \Theta(v,r) = \Theta_0(v)+ \frac{2\kappa\sigma(v)r^3}{3n_0(v)},
\end{equation}
where $\Phi_0$ and $\Theta_0$ are functions determined by the initial/boundary data. Now, for the Lagrangian labels, we can choose 
\begin{equation}
    \mathrm{a}^1 = v, \quad \mathrm{a}^2=\theta, \quad \mathrm{a}^3 = \phi,
\end{equation}
as these variables are constant along the flow of $u$. Because the variables $b_i$ obey the same conservation equations that $n$ and $s$ obey, it follows that $b_i(v,r) = b_i^{(0)}(v)/r^2$. The functions $b_i^{(0)}$ are also constrained by the velocity equation in \eqref{eqn:velocity-rel}. Indeed, a direct calculation shows that $\mathcal{C}=-fu^\flat$ if and only if 
\begin{equation}
    b_1^{(0)}(v) = -n_0(v)\left(\Phi_0'(v) + \sqrt{\frac{\rho_* + q v}{\kappa}}\,\Theta_0'(v)\right), 
    \quad b_2^{(0)}=b_3^{(0)}=0.
\end{equation}

This completes the construction of the solution, provided the boundary conditions for the velocity potentials $\Phi_0,\Theta_0$, the Lagrangian coordinates $\mathrm{a}^i$, and the particle density $n_0$. Note that this solution satisfies the null, weak and dominant energy conditions in the regions $r>0$ and $\rho_* + qv>0$. Because $\rho>0$, $P<0$, and $f<0$, for every future-pointing null vector $k$ and future-pointing unit timelike vector $w$, we have
\begin{equation}
    T_{\alpha\beta}\, k^\alpha k^\beta = -f \left(g(u,k)\right)^2\geq 0, 
    \quad 
    T_{\alpha\beta}\, w^\alpha w^\beta = \rho -f \left(g(u,w)\right)^2>0,
\end{equation}
proving the validity of the null and weak energy conditions, respectively. Finally, since the flux vector corresponding to $w$ is given by 
\begin{equation}
    {}^wJ^\alpha:=-T^\alpha{}_\beta w^\beta = \rho w^\alpha + fg(u,w)u^\alpha,
\end{equation}
direct inspection shows that ${}^wJ$ is future-pointing and causal, proving that the dominant energy condition is satisfied.

\section{Discussion}

We have presented a geometric formulation, in terms of differential forms, of the constrained action principle for relativistic perfect fluids. For timelike flows, the construction reproduces the standard description of Schutz, Brown and others. Keeping the thermodynamic and kinematical variables separate allows the null extension to be examined without recovering the density and velocity from the norm of the particle current. Although independent variations of the density and velocity are already discussed by Brown, this choice is particularly useful here because the usual reconstruction from the current becomes degenerate when the current is null. While the matter equations imply $h=\rho+P=0$, restricting the states realised by null solutions, requiring that  changes of the normalization of the velocity $u$ to be a symmetry of the complete action imposes a stronger condition: that energy density $\rho$ is a function of the entropy per particle alone, \textit{i.e.} $\rho=\rho(s/n)$, for which vanishes $\rho+P$ identically. The resulting stress-energy tensor contains a metric-proportional contribution and a rank 1 null contribution. Its conservation implies that the pressure is constant along each flow line, although it may vary between flow lines:
\begin{equation}
    \pounds_u P=0,\quad  \nabla_\alpha P=fu^\beta \nabla_\beta u_\alpha+\nabla_\beta (f u^\beta)u_\alpha.
\end{equation}
Pressure gradients can therefore produce departures from pre-geodesic motion. When $dP=0$ and $f\neq0$, the flow lines are null pre-geodesics. Setting $\rho=P=0$, one recovers the null dust action principle derived by Bi{\v c}{\'a}k and Kucha{\v r}, with non-negative null-dust density corresponding to $f\leq0$. A constant, non-zero $\rho$ adds a constant vacuum energy-like contribution.

The thermodynamic interpretation is not straightforward and requires further consideration. Setting $\rho(n,s)=\varrho(s/n)$ and assuming the validity of the usual Maxwell relations throughout, we get $T=\varrho'(s/n)/n$ and $\mu=-sT/n$. Consequently, for positive $n,s$, a positive temperature requires a negative chemical potential. Although a negative chemical potential is not itself inconsistent with ordinary thermodynamics, the usual equilibrium stability condition imposes a stronger restriction. If $n,s$ are interpreted as independently variable thermodynamic densities, the Hessian of $\rho(n,s)$ must be positive semidefinite. In particular, its determinant must be non-negative. Here, however,
\begin{equation}
    \det 
    \begin{pmatrix}
        \rho_{nn} & \rho_{ns} \\
        \rho_{sn} & \rho_{ss}
    \end{pmatrix}
    = -\frac{(\varrho'(s/n))^2}{n^4} = -\frac{T^2}{n^2}. 
\end{equation}
For $T\neq 0$, the failure of the usual stability makes it difficult to interpret this class of systems as usual thermodynamic systems. It's perhaps best that one should interpret these null fluids as an effective description of null radiation interacting with a scalar contribution that behaves as vacuum energy. The spherically symmetric solution constructed in Section \ref{sec:null-example}, which locally satisfies the weak, null and dominant energy conditions, provides a concrete realization of this description. It would be interesting to investigate whether such a model can emerge as an effective description of a more fundamental theory, such as a quantum field theory, and whether that derivation supplies a physical interpretation of the scalar variables. However, this question lies beyond the scope of the present work.

\section*{Acknowledgements}

The author would like to thank Bernhard Schutz and Bruce Allen for helpful discussions.

\appendix

\section*{Appendix}

\section{Construction of Lagrangian coordinates}
\label{sec:lag-coord-proof}

The standard theorems on the existence and uniqueness of solutions to ordinary differential equations (\textit{c.f.} \cite{arnold1992odes}) imply that, for every point $q\in\mathscr{M}$, there is an open neighbourhood $\mathscr{V}\subset\mathscr{M}$ of $q$, $\varepsilon>0$, and a family of local diffeomorphisms $\{\phi_{\tau}\}_{|\tau|<\varepsilon}$ of $\mathscr{V}$ such that 
\begin{equation}
\begin{dcases}
    \frac{\partial}{\partial \tau}\phi_{\tau}(p) = u_{\phi_{\tau}(p)} & p\in\mathscr{V},\quad |\tau|<\varepsilon 
    \\ \\ 
    \phi_0 = \textrm{id}
\end{dcases}
\end{equation}
Moreover, uniqueness implies 
\begin{equation}
    \phi_{\tau+s} = \phi_{\tau}\circ\phi_s,\quad \phi_{\tau}^{-1}=\phi_{-\tau},
\end{equation}
whenever $\tau,s,\tau+s\in(-\varepsilon,\varepsilon)$. Now, fix a point $q\in S$, and choose $S$ to be small enough so that $S\subset\mathscr{V}$. We define $F:(-\varepsilon,\varepsilon)\times S\to\mathscr{M}$ by 
\begin{equation}
    F(\tau,p) = \phi_{\tau}(p),\quad (\tau,p)\in(-\varepsilon,\varepsilon)\times S. 
\end{equation}
We shall show, after a possible shrinking of $\varepsilon$ and $S$, that $F$ is a diffeomorphism between $(-\varepsilon,\varepsilon)\times S$ and a neighbourhood $\mathscr{U}$ of $q$. Let $w\in T_qS$, and $\gamma:(-\delta,\delta)\to S$ be a $C^1$ curve such that $\gamma(0)=q$ and $\dot{\gamma}(0)=w$. Then, the differential of $F$ along $w$ is given by 
\begin{align}
    dF_{(0,q)} w = \frac{\partial}{\partial\tau}\bigg\rvert_{\tau=0}\phi_0(\gamma(\tau))
        = \dot{\gamma}(0) = w. 
\end{align}
Moreover, because
\begin{equation}
    \frac{\partial}{\partial \tau}\bigg\rvert_{\tau=0} F(\tau,q) = 
    \frac{\partial}{\partial \tau}\bigg\rvert_{\tau=0} \phi_{\tau}(q) = u_q,
\end{equation}
and $u_q$ is non-zero and causal -- hence transversal -- to $T_q S$, it follows that $dF_{(0,q)}$ is a linear isomorphism. 

Thus, the inverse function theorem implies that $F$ is a diffeomorphism mapping $(-\varepsilon,\varepsilon)\times S$, for some $\varepsilon>0$ and an open neighbourhood of $q$ in $S$, onto an open neighbourhood $\mathscr{U}\subset\mathscr{M}$. Consequently, for every $p\in\mathscr{U}$, there is a point $p_0\in S$ and a number $\tau(p)\in(-\varepsilon,\varepsilon)$ such that 
\begin{equation}
    F^{-1}(p) = (\tau(p),p_0).
\end{equation}
Without loss of generality, we can assume that $S$ is covered by a local coordinate system $x^1,x^2,x^3$. We define the functions $\tau,\mathrm{a}^1,\mathrm{a}^2,\mathrm{a}^3$ on $\mathscr{U}$ by 
\begin{subequations}
\begin{align}
    & \tau = \pi_1\circ F^{-1}, &&
    \\
    & \mathrm{a}^i = x^i\circ\pi_2\circ F^{-1}, && 
    1\leq i\leq 3,
\end{align}    
\end{subequations}
where $\pi_1$ and $\pi_2$ are the projections of $(-\varepsilon,\varepsilon)\times S$ to $(-\varepsilon,\varepsilon)$ and $S$ respectively. The function $(\tau,\mathrm{a}^i)$ is clearly a coordinate system, as a composition of diffeomorphisms. Thus, to conclude the proof, we need to show that the maps $\tau$ and $\mathrm{a}^i$ satisfy $\pounds_u \tau = 1$ and $\pounds_u\mathrm{a}^i = 0$, for all $i\in\{1,2,3\}$. For each $p\in\mathscr{U}$, there is $\delta>0$ such that $\phi_s(q)\in\mathscr{U}$ whenever $|s|<\delta$. Hence, for sufficiently small $s>0$, we have 
\begin{align}
    \tau(\phi_s(p)) & = \pi_1\circ F^{-1}\circ F(s + \tau(p),p_0)
    \nonumber\\
    & = \pi_1(s+\tau(p),p_0)
    \nonumber\\
    & = s + \tau(p),
\end{align}
and therefore 
\begin{equation}
    \pounds_u \tau(p) = \frac{d}{ds}\bigg\rvert_{s=0}\tau(\phi_s(p)) = 1.
\end{equation}
Similarly, for every $i\in\{1,2,3\}$ we have 
\begin{align}
    \mathrm{a}^i(\phi_s(p)) & = x^i\circ \pi_2\circ F^{-1}\circ F(s + \tau(p),p_0)
    \nonumber\\
    & = x^i\circ\pi_2(s+\tau(p),p_0) 
    \nonumber\\
    & = x^i(p_0) = \mathrm{a}^i(p).
\end{align}
which in turn yields 
\begin{equation}
    \pounds_u\mathrm{a}^i(p) = \frac{d}{ds}\bigg\rvert_{s=0}\mathrm{a}^i(\phi_s(p)) = 0.
\end{equation}

\hfill$\Box$ 

\section{Variations of differential forms}
\label{sec:tetrad-formalism}

\renewcommand{\thetheorem}{\Alph{section}.\arabic{theorem}}
\setcounter{theorem}{0}

Here, we establish some rules on calculating the variations of general differential forms. Under a variation of the tetrad $\{e^\alpha\}$, the variations of the components of a differential form are not necessarily zero. If $\xi$ is an arbitrary $k$-form on $\mathscr{M}$, its variation $\delta_{\mathbf{e}}\xi$ is given by 
\begin{align}
    \label{eqn:var-k-form}
    \delta_\mathbf{e}\xi = \overline{\delta}_\mathbf{e}\xi + \delta e^{\alpha}\wedge i_{e_\alpha}\xi,
\end{align}
where 
\begin{equation}
    \overline{\delta}_{\mathbf{e}}\xi := \frac{1}{k!}\delta \xi_{\alpha_1...\alpha_k}
    e^{\alpha_1}\wedge ... \wedge e^{\alpha_k}. 
\end{equation}

\begin{lemma}
Let $\mu_{\alpha_1...\alpha_k}$ be the tensor-valued differential forms defined by 
\begin{equation}
    \mu_{\alpha_1...\alpha_k} = \textup{Vol}_g(e_{\alpha_1},...,e_{\alpha_k},\cdot). 
\end{equation}
Then: 
\begin{subequations}
\begin{align}
    & \delta_{\mathbf{e}}\mu_{\alpha_1...\alpha_k} = \delta e^{\beta}\wedge\mu_{\alpha_1...\alpha_k\beta}, 
    \label{eqn:var-vol-form}\\
    \nonumber\\
    & \mu^{\alpha_1...\alpha_k} = \star( e^{\alpha_1}\wedge...\wedge e^{\alpha_k}),
    \label{eqn:vol-star-prod}
\end{align}    
\end{subequations}
where indices are raised by the inverse metric tensor.
\end{lemma}

\begin{proof}
Eq. \eqref{eqn:var-vol-form} follows as a direct consequence of \eqref{eqn:var-k-form}. To prove \eqref{eqn:vol-star-prod}, recall that the Hodge star product operator maps every $k$-form $\xi$ to a $(n-k)$-form $\star\xi$, where $n:=\dim\mathscr{M}$, whose components are given by 
\begin{equation}
    \label{eqn:hodge-star-def}
    (\star\xi)_{\beta_{k+1}...\beta_n} = \frac{1}{k!}\xi_{\beta_1...\beta_k}\mu^{\beta_1...\beta_k}{}_{\beta_{k+1}...\beta_n}
\end{equation}
If $\xi = e^{\alpha_1}\wedge...\wedge e^{\alpha_k}$, then 
\begin{equation}
    \xi_{\beta_1...\beta_k} = \sum_{\sigma\in\mathcal{S}_k} \textrm{sgn}\,\sigma\,
    {\delta^{\alpha_1}}_{\beta_{\sigma(1)}} ... {\delta^{\alpha_k}}_{\beta_{\sigma(k)}}.
\end{equation}
Consequently, 
\begin{align}
    \star\xi & = \frac{1}{k!}\sum_{\sigma\in\mathcal{S}_k}
    \textrm{sgn}\,\sigma\,
    {\delta^{\alpha_1}}_{\beta_{\sigma(1)}} ... {\delta^{\alpha_k}}_{\beta_{\sigma(k)}}
    \mu^{\beta_{1}...\beta_k}
    \nonumber\\
    & = \frac{1}{k!}\sum_{\sigma\in\mathcal{S}_k}
    {\delta^{\alpha_1}}_{\beta_{\sigma(1)}} ... {\delta^{\alpha_k}}_{\beta_{\sigma(k)}}
    \mu^{\beta_{\sigma(1)}...\beta_{\sigma(k)}}\nonumber\\
    & = \mu^{\alpha_1...\alpha_k}.
\end{align}
\end{proof}

\begin{lemma}
If $\xi$ is a $k$-form on $\mathscr{M}$, then 
\begin{equation}
    \delta_\mathbf{e}\star\xi = \star\delta_\mathbf{e}\xi + \delta e^\alpha\wedge i_{e_\alpha}\star\xi - \star\left(\delta e^\alpha\wedge i_{e_\alpha}\xi \right). 
    \label{eqn:xi-var-arb}
\end{equation}
In particular, if $\delta_\mathbf{e}\xi = 0$, then 
\begin{equation}
    \delta_\mathbf{e}\star\xi = \delta e^\alpha\wedge i_{e_\alpha}\star\xi - \star\left(\delta e^\alpha\wedge i_{e_\alpha}\xi \right).
    \label{eqn:xi-var-inv}
\end{equation}
\end{lemma}

\noindent
\textbf{Remark.} Recall that $i_X\omega = \omega(X,\cdot)$, for every differential form $\omega$.  

\begin{proof}
For every $k$-form $\xi$ on $\mathscr{M}$, we have 
\begin{align}
    \delta_\mathbf{e}\star\xi & = \delta_\mathbf{e}\left(\frac{1}{k!}\xi_{\alpha_1...\alpha_k}\mu^{\alpha_1...\alpha_k}\right)
    \nonumber\\
    & = \frac{1}{k!}\left(\delta_\mathbf{e}\xi_{\alpha_1...\alpha_k}\mu^{\alpha_1...\alpha_k} + \xi^{\alpha_1...\alpha_k}\delta\mu_{\alpha_1...\alpha_k}\right)
    \nonumber\\
    & = \star\overline{\delta}_\mathbf{e}\xi + \frac{1}{k!}\xi^{\alpha_1...\alpha_k}\delta e^\beta\wedge\mu_{\alpha_1...\alpha_k\beta} 
    \nonumber\\
    & = \star\overline{\delta}_\mathbf{e}\xi + \delta e^\alpha\wedge(i_{e_\alpha}\star\xi).
    \label{eqn:arb-var-proof}
\end{align}
Eq. \eqref{eqn:xi-var-arb} then follows by substituting eq. \eqref{eqn:var-k-form} to \eqref{eqn:arb-var-proof}.
\end{proof}

%

%\raggedbottom

\end{document}